\documentclass[a4paper,12pt]{article}
 \usepackage[margin=2.5cm]{geometry}
 \usepackage{enumerate}
 \usepackage{bbm}

\usepackage[utf8]{inputenc} 
\usepackage[bookmarks=true,pagebackref=false, colorlinks, breaklinks]{hyperref}
\hypersetup{linkcolor=blue,citecolor=blue,filecolor=black,urlcolor=blue}
\usepackage{amsmath,amssymb,amsfonts,amsthm}
\usepackage{color}    

\theoremstyle{plain}
\newtheorem{theo}{Theorem}[section]
\newtheorem{lemm}[theo]{Lemma}
\newtheorem{prop}[theo]{Proposition}
\newtheorem{coro}[theo]{Corollary}
\theoremstyle{definition}
\newtheorem{defi}[theo]{Definition}
\theoremstyle{remark}

\newtheorem{remark}[theo]{Remark}

\newcommand{\compcent}[1]{\vcenter{\hbox{$#1\circ$}}}
\newcommand{\comp}{\mathbin{\mathchoice
{\compcent\scriptstyle}{\compcent\scriptstyle}
{\compcent\scriptscriptstyle}{\compcent\scriptscriptstyle}}} 
\newcommand{\inner}  [2]{\langle #1 , #2 \rangle}
\newcommand{\binner} [2]{\big\langle #1 , #2 \big\rangle}

\newcommand{\Svn}{S_{\rm vN}}
\newcommand{\Mob}{{\rm M\ddot{o}b}}

\newcommand{\Diff}{{\rm Diff}_+(S^1)}

\def\RR{{\mathbb R}}
\def\CC{{\mathbb C}}
\def\NN{{\mathbb N}}

\def\A{{\mathcal A}}
\def\B{{\mathcal B}}

\def\H{{\mathcal H}}

\def\K{{\mathcal K}}

\def\R{{\mathcal R}}

\def\f{\varphi}
 \def\l{\lambda}
\def\R{{\mathcal R}}
\def\s{\sigma}
\def\Ad{{\hbox{\rm Ad\,}}}
\def\dim{\mathrm{dim}\,}
 \def\1{{\mathbbm 1}}
 \def\<{\langle}
 \def\>{\rangle}

\DeclareMathOperator{\Tr}{Tr}
 \newcommand{\im}{\mathrm{Im}\,}

\title{Towards entanglement entropy with UV-cutoff in conformal nets}
\date{}
\author{
{\bf Yul Otani} \\
Graduate School of Mathematical Sciences, The University of Tokyo\\
3-8-1 Komaba Meguro-ku Tokyo 153-8914, Japan.\\
 email: {\tt otaniyul@gmail.com}
   \vspace{0.5cm}\\
   {\bf Yoh Tanimoto} \\
   Dipartimento di Matematica, Universit\'a di Roma Tor Vergata\\
   Via della Ricerca Scientifica, 1, I-00133 Roma, Italy\\
   email: {\tt hoyt@mat.uniroma2.it}
}
\begin{document}
\maketitle
\begin{abstract}
We consider the entanglement entropy for a spacetime region and its spacelike complement
in the framework of algebraic quantum field theory.
For a M\"obius covariant local net satisfying either a certain nuclearity property or the split property,
we consider the von Neumann entropy for type I factors between local algebras
and introduce an entropic quantity.
Then we implement a cutoff on this quantity with respect to the conformal Hamiltonian
and show that it remains finite as the distance of two intervals tends to zero.
We compare our definition to others in the literature.
\end{abstract}
\section{Introduction}\label{sec:intro}

Recently, entanglement entropy in quantum field theory has been of remarkable interest in various context,
e.g.\! in relation with the black hole entropy \cite{bkls86, CW94, srednicki93} and holographic principle \cite{ryutak06}.
Concrete formulas for the entanglement entropy in various models have been conjectured and their behaviors have been
investigated \cite{cashue09}.
Actually, the definition of entanglement entropy itself is not straightforward.
Let us look at this issue closely.

In quantum mechanics, one considers a system on a Hilbert state $\H$ composed of
the subsystem $1$ with Hilbert space $\H_1$ and its complementary subsystem $2$ with Hilbert space $\H_{2}$,
so that $\H = \H_1\otimes \H_{2}$.
Then one takes a state $\f$ on the full system $\B(\H)$ and restricts it to the subalgebra
$\B(\H_1) \cong \B(\H_1)\otimes \CC\1_{\H_{2}}$.
This corresponds to taking the global density matrix $\rho_\f$ and ``tracing out'' the inaccessible degrees of freedom
of the subsystem $2$ with the partial trace $\Tr_{\H_{2}}(\rho_\f)$.
The von Neumann entropy of the resulting reduced density matrix is then called
the \emph{entanglement entropy} of the global state with respect to the subsystem $1$, namely,
\[
  S_1(\rho) := \Svn\left(\Tr_{\H_{2}}(\rho_\f)\right).
\]

In quantum field theory (QFT), defined on $\RR^d$, it is natural to consider the division of the whole system
into a spacetime region $O$ and its spacelike complement $O'$.
However, the corresponding factorization of the Hilbert space $\H = \H_O\otimes\H_{O'}$ does not exist, even
for the free fields: the local algebras of observables are not of type I in the sense of von Neumann
(see \cite{Araki64} for the free fields, \cite{bdf87} for a class of QFT with a nuclearity condition
and \cite{BGL93} for general conformal field theory),
and partial trace does not make sense.
Now, whereas states on type I algebras may have finite entropy, a
reasonable generalization of entropy is infinite for \emph{any} state on
an algebra of type III$_1$.  It should be stressed that this
divergence is different from the conventional explanation that the
entanglement entropy of a state is infinite because of increasing
contribution of fluctuations of high-energy states and therefore needs
an UV-cutoff. Hence, we need to deal with two issues: non-type I local
algebras and UV-divergence.

In physics literature, this is done by the lattice regularization:
then the local algebras are of type I (often just finite dimensional), and
there are not too many high-energy states, and indeed the entanglement entropy of the ground state
has been estimated in several cases and some exact results have been claimed (see \cite{calcar04} for an overview of results).

From a mathematical point of view, however, the lattice regularization is not completely satisfactory.
For one thing, there may be many different lattice regularizations for a given quantum field theory,
and it is unclear how intrinsic the value obtained in a particular regularization is.
For another, situations of interest in various contexts, e.g.\! of black hole entropy, computations on lattices are not
quite available. Therefore, we would like a direct approach to entanglement entropy based on the continuum.

The notion of entropy is most clearly presented in terms of operator algebras (see e.g.\! \cite{ohyapetz}),
hence for this purpose, the operator-algebraic approach, or algebraic QFT \cite{haag,araki} seems to be the most appropriate.
The main ingredient of the theory is a net of von Neumann algebras $\{\A(O)\}$ indexed by (bounded) regions of the spacetime,
of which each element $\A(O) \subset \B(\H)$ is the algebra generated by the observables which can be measured in $O$,
and they are subject to several conditions called the Haag-Kastler axioms.
The entanglement of the vacuum state has been studied in this framework since the discovery of Bell's inequality
\cite{SW85, SW87} and related with basic properties such as the Reeh-Schlieder property, see \cite{Yngvason15} for a review.
On the other hand, defining entropy in QFT is already a nontrivial task even for the free field, hence the usual caveat
that no interacting QFT in the physical four dimensions has been constructed in the Haag-Kastler framework is not a main concern in this respect.

As remarked above, local algebras in AQFT are not of type I, hence the notion of entanglement entropy must be
treated with care. On the other hand, let us point out that in many examples of Haag-Kastler nets,
the split property holds: there are type I algebras between local algebras of two regions, one included
in the other with finite distance. And the split property is assured when a nuclearity condition holds.
A nuclearity condition, although there are several variations, roughly says that ``the number of states''
in a region with restricted energy is small. It should be noted that ``the number of states'' is not a
precise expression, but the first of such notions, the compactness condition, has been formulated
based on this idea \cite{haaswi65} and various nuclearity conditions put certain quantitative restriction
on the state space. It is natural to expect that they are useful in the consideration of entanglement entropy with
UV-cutoff.

Indeed, \cite{narnhofer94}, Narnhofer
gave a slightly altered definition of localized entropy of
the vacuum $\omega$ restricted to an local algebra $\A(O)$, using an auxiliary quantity $\delta > 0$
as a regularizing spatial parameter (in the sense that, when $\delta =0$,
the author's definition would recover the standard, yet divergent definition).
Furthermore, it was shown that if the local net satisfies a stronger version of the nuclearity condition of Buchholz and Wichmann
\cite{bucwic86}, then one could take $O$ growing to the whole Minkowski space,
with $\delta$ growing also appropriately, and the entropy would tend to zero, as expected from the vacuum state.
Further relations between nuclearity conditions and entropy have been investigated in \cite{narnhofer02}.
Another recent result exploiting a nuclearity condition is available in \cite{HS17},
where entanglement entropy between distant regions is estimated without the need of cutoff regularizations.

Yet, the results claimed by physicists often concern the entanglement entropy between
a region and its spacelike complement, hence they are in contact, and a UV-cutoff is necessary.
In this paper, we try to make sense of it in the framework of AQFT.
In order to concentrate on the implementation of a cutoff, we consider a chiral component
of a two-dimensional conformal field theory and its extension to the circle $S^1$,
the one-point compactification of the spacetime $\RR$.
The conformal Hamiltonian $L_0$ has a discrete spectrum, which makes our analysis easier.

\subsubsection*{Outline of this work}
If a conformal Haag-Kastler net (conformal net for short) satisfies the so-called \emph{split property},
even though local algebras $\A(I)$ are of type ${\rm III_1}$, we can ``approximate it from the outside'' by factor of type I.
More concretely, we take a spacing parameter $\delta>0$ and an interval $I_\delta$ which is obtained by expanding $I$ by $\delta$.
The split property then assures the existence of a factor $\R$ of type I between $\A(I)$
and $\A(I_\delta)$. As these algebras $\R$ are of type I, the von Neumann entropy can be defined for any state,
although we cannot exclude the possibility that they are infinite for interesting states such as the vacuum.
Yet, we can consider the infimum of the von Neumann entropies of all states which dominates the given state
restricted to $\A(I)\vee \A(I_\delta')$. This resembles the definition of entanglement entropy when
two algebras do not generate the whole $\B(\H)$, and in our case, turns out to be finite under a certain growth condition on the eigenspaces of the conformal Hamiltonian $L_0$. 
Next we consider, while $\delta$ is still present, a new cutoff parameter $E$.
This cutoff parameter is then used to regularize our considered states by cutting off the contributions of ``conformal energy higher than $E$''.
Since the conformal Hamiltonian $L_0$ has a discrete spectrum, many of the calculations can be simplified,
and we acquire an upper bound for the entropy that is independent of $\delta$.
We think that this captures a certain aspect of entanglement entropy with a UV-cutoff $E$.
Actually, the technical prerequisites for the finiteness of the former are shown to be less restrictive
than for its geometric counterpart: it is only required that the chiral net satisfies the split property.
Our main result is then the finiteness of this cutoff quantity without the regularization by distance $\delta$.
As the split property implies the modular compactness \cite{bdl90a}, the result that the cutoff gives
a finite entropy seems to make sense.

This paper is organized as follows.
In Section \ref{sec:math}, we present the mathematical tools to be used,
including a brief review of M\"obius covariant local nets and of von Neumann entropy.
In particular, we recall the defining properties of our nets of observables in Section \ref{subsec:net}
and list important assumptions on nuclearity in Section \ref{subsec:lp}.
In Section \ref{sec:cut}, we introduce our definition of regularized entropic quantity for a M\"obius covariant local net
satisfying the split property, and prove its finiteness (see Theorem \ref{thm:cut3}) given that the net satisfies the condition
\ref{def:ea}(\ref{it:1n}) of conformal nuclearity. We finish with our conclusions in Section \ref{sec:conclusions}.

This work has been carried out as a part of the Ph.D.\! project of the author (Y.O.) \cite{Otani17}.


\section{Mathematical preliminaries}\label{sec:math}
In this Section, we introduce our mathematical framework.
Section \ref{subsec:net} recalls M\"obius covariant local nets and their basic properties.
In \ref{subsec:Svn}, we review basic properties of von Neumann entropy.
Notions concerning nuclear maps and their use in AQFT are summarized in Sections \ref{subsec:lp}, \ref{subsec:nuc},
as well as other related conditions (Definition \ref{def:ea}).


\subsection{M\"obius covariant local nets}
\label{subsec:net}
We consider chiral components of conformal field theories in two-spacetime dimensions.
Several important observables such as the stress-energy tensor or currents in a two-dimensional conformal field theory
decompose into chiral components defined on the lightrays, and each component can be studied separately.
The Poincar\'eé group restricted to one lightray is the the translation-dilation group.
Furthermore, by conformal covariance, the theory defined on the lightray $\RR$ extends to its one point compactification $S^1$,
and it is covariant under an action of the M\"obius group $\Mob \cong \mathrm{PSU}(1,1)$.

In the operator-algebraic approach, we are concerned with the algebras of observables associated with
local regions, and they are (non empty, non dense, open and connected) intervals on $S^1$.
These algebras of observables are required to satisfy a standard set of properties (the Haag-Kastler axioms),
which we summarize below.

We denote by $\mathcal{I}$ the set of non-empty, non-dense, connected open intervals of $S^1$.
For $I\in\mathcal{I}$, we denote by $I'$ its causal complement i.e.\! the interior of $S^1\setminus I$.
The {\bf distance} between two intervals $I_1,I_2\in\mathcal{I}$ is their angular distance,
i.e.\! the infimum of a value $|\theta|$ such that $e^{i\theta}I_1$ intersects $I_2$
(considering $S^1$ as a subset of $\CC$).
Also, for two intervals $I_1,I_2\in\mathcal{I}$, we say that $I_1 \Subset I_2$ if the closure $\overline{I_1}$ is contained in $I_2$, that is, if $I_1$ and $I_2'$ have a positive distance.

A {\bf M\"obius covariant local net} consists of a quadruple $(\A,U,\Omega,\H)$,
where $\H$ is the Hilbert space of the theory, $\Omega\in\H$ is a unit vector corresponding to the vacuum state,
$U$ is a strongly continuous unitary representation of $\Mob$ on $\H$, and $\A$
is a family of von Neumann algebras acting on $\H$ and indexed by elements of $\mathcal{I}$.
Those are supposed to satisfy the following properties.
  \begin{enumerate}
  \item {\bf Isotony.} For any pair of intervals $I_1, I_2 \in \mathcal{I}$, if $I_1 \subset I_2$, then $\A(I_1) \subset \A(I_2)$.
  \item {\bf Locality.} For any pair of intervals $I_1, I_2 \in \mathcal{I}$, if $I_1 \cap I_2 = \emptyset$,
  then $[\A(I_1),\A(I_2)]=0$.
  \item {\bf Covariance.} For $I\in\mathcal{I}$ and $g \in \Mob$,
  it holds that $\Ad U(g)(\A(I)) = \A(gI)$.
  \item {\bf Positivity.} The generator $L_0$ of the rotation one-parameter group $\{U(\rho_\theta) = e^{i\theta L_0}\}_{\theta\in\RR}$ has a positive spectrum.
  \item {\bf Uniqueness of the vacuum.} $\Omega$ is the unique (up to phase) unit vector in $\H$
  which is invariant for $U$.
  \item {\bf Cyclicity of the vacuum.} $\Omega$ is cyclic for the algebra
  $\bigvee_{I\in\mathcal{I}}\A(I)$.
  \end{enumerate}
These assumptions are standard, see e.g.\! \cite{gabfro93, frejor96}
and \cite{Rehren15, kawahigashi15} for recent reviews.
From these axioms, some properties automatically follows.
\begin{itemize}
\item {\bf Discrete spectrum of $L_0$.} It holds that ${\rm sp}(L_0) \subset \mathbb{N}$.

\item {\bf Reeh-Schlieder property.}\footnote{This property shows aspects of non-independence of the vacuum state,
as any other vector state can be approximated by local operations on it, see \cite{SW85, SW87}
for the relation between the Reeh-Schlieder property and the violation of Bell's inequality.
Another consequence is the impossibility of existence of an local number operator (as it would have the vacuum vector as an eigenvector).}
The vacuum vector $\Omega$ is cyclic for any local algebra $\A(I)$, for $I\in\mathcal{I}$.

\item {\bf Haag duality:} For any interval $I\in\mathcal{I}$, it holds that $\A(I)' = \A(I')$.

\item {\bf Additivity:} If $\{I_n\in \mathcal{I}\}_{n\in\mathbb{N}}$ is a covering for an interval
$I \subset \cup_{n\in\mathbb{N}} \, I_n \in \mathcal{I}$, then $\A(I) \subset \bigvee_{n\in\mathbb{N}}\,\A(I_n)$.
  
\item {\bf Factoriality.} Local algebras are factors of type ${\rm III}_1$.
  
\end{itemize}

These are the most general assumptions, and examples satisfying them are not necessarily ``physical''.
A pathological example is the infinite tensor product of any given M\"obius covariant local net,
which fails to have the stress-energy tensor \cite[Section 6]{CW05}.
On the other hand, a M\"obius covariant local net with a natural stress-energy tensor,
and hence conformal (diffeomorphism) covariance, satisfies the modular compactness condition \cite{bdl90a}:
it follows from the conformal covariance through the split property \cite{mtw16}.
Modular compactness says that state space is ``small'' in a certain sense.
Indeed, many examples studied in the physical literature have the corresponding M\"obius covariant local nets
with a strengthened state space property.
We discuss this issue briefly at the end of Section \ref{subsec:nuc}.


\subsection{Von Neumann entropy}
\label{subsec:Svn}

Here we make a brief review on von Neumann entropy to set the notation and basic properties used later in this work
(for a nice exposition, see e.g.\! \cite{ohyapetz}). Let $\H$ be separable Hilbert space, and $\B(\H)$
the algebra of bounded linear operators on $\H$.
Recall that any normal state $\f$ on $\B(\H)$ has an associated positive, normalized and trace class operator
(``the density matrix'') $\rho_\f \in \B(\H)$, such that $\f = \Tr(\rho_\f \,\cdot\,)$, where $\Tr$ is the (non-normalized) trace functional.

\begin{prop}\label{pro:etalp}
  For a parameter $p$ with $0<p<1$, there is a constant $c_p>0$ such that
  \[ -t\log t \le c_p t^p \qquad (t \ge 0),\]
  where the left-hand side equals $0$ at $t = 0$ by convention.
  Moreover, the optimal value is $c_p = \frac{1}{(1-p)e}$.
\end{prop}
\begin{proof} By elementary calculus, the differentiable function given by
\[t\in \RR_{\ge 0} \mapsto -t^{1-p}\log t \in \RR \]
attains its maximum at $t_0=e^{-1/(1-p)}$, with value $c_p = \frac1{(1-p)e}$.
By multiplying the inequality $-t^{1-p}\log t \le c_p$ by $t^p$ we obtain the claimed inequality.
\end{proof}

Let $\f$ be a normal state on $\B(\H)$, and $\rho_\f$ its associated
density matrix.  Then its {\bf von Neumann entropy} is defined as
$\Svn(\f) := -\Tr(\rho_\f\log(\rho_\f))$.

\begin{prop}\label{pro:Svn}
  Let $\f$ be a normal state on $\B(\H)$.
  Then $\Svn(\f)$ has the following properties.
  \begin{enumerate}
  \item Positivity: $0 \le \Svn(\f) \le \log(\dim(\H))$. Moreover, $\Svn(\f) = 0$ if and only if $\f$ is pure.
  \item Invariance: If $\s$ is a $*$-automorphism of $\B(\H)$, then $\Svn(\f\comp \s) = \Svn(\f)$.
  \item Concavity: If $\f = \sum_k \l_k \f_k$ is a convex decomposition of $\f$ (\textit{i.e.}\! $\{\f_k\}_{k\in\mathbb{N}}$ if a family of normal states on $\B(\H)$ and $\{ \lambda_k \ge 0 \}_{k\in\mathbb{N}}$ if a family of positive numbers such that $\sum_{k\in\mathbb{N}}\lambda_k = 1$), then
    \[ \sum_{k\in\mathbb{N}} \lambda_k \Svn(\f_k)
    \le \Svn\left(\sum_{k\in\mathbb{N}} \lambda_k \f_k\right) \le
    \sum_{k\in\mathbb{N}} \lambda_k \Svn(\f_k) - \sum_{k\in\mathbb{N}} \lambda_k\log \l_k .\]
  \end{enumerate}
\end{prop}
\begin{proof}Positivity and invariance are straightforward from the definition. For a proof of concavity, see \cite[Proposition 1.6 and 6.2]{ohyapetz}. 
\end{proof}

We rephrase the concavity property in the following corollary, which will be used later in our calculations.
\begin{coro}\label{pro:Svnsum} Let $\{\f_k\}_{k\in\mathbb{N}}$ be a family of pure states on $\B(\H)$, and $\{ \lambda_k \ge 0 \}_{k\in\mathbb{N}} \in l^1_+(\mathbb{N})$ be a summable sequence of positive parameters. Define the positive functional $\f$ acting on $\B(\H)$ by $\f := \sum_{k\in\mathbb{N}} \lambda_k \f_k$. Clearly, its norm is given by $\|\f\| = \sum_{k\in\mathbb{N}}\lambda_k$. Then, the entropy of the state $\f/\|\f\|$ satisfies the following inequality
  \[ \Svn\left(\frac{\f}{\f(\1)}\right) \le \log(\f(\1)) - \frac{1}{\f(\1)}\sum_{k\in\mathbb{N}}\lambda_k\log\lambda_k. \]
\end{coro}


The definition of the von Neumann entropy relies on the underlying Hilbert space, where normal positive functionals can be associated with
density matrices. However, by the properties of Proposition \ref{pro:Svn}, one can easily define the von Neumann entropy for states on an abstract type I factor by the following.

\begin{defi}\label{def:SvnI} Let $\R$ be a factor of type I.
  Then there is a Hilbert space $\mathcal{K}$ such that a $*$-isomorphism $\s : \B(\K) \to \R$ exists.
  Then, for any normal positive functional $\f$ on $\R$, its von Neumann entropy $S_\R(\f)$ is defined by
  \[S_\R(\f) := \Svn(\f\comp \s). \] 
\end{defi}

Since the invariance property in Proposition \ref{pro:Svn} holds,
the above definition is independent on the choice of $\s$, and thus is well-defined.


For general von Neumann algebras (such as local algebras, which are in many cases factors of type III$_1$),
there might not be corresponding density matrices nor traces, so the usual definition of von Neumann entropy does not make sense.
An alternative definition by means of relative entropy is explained in \cite[Chapter 6]{ohyapetz},
and shown to diverge for any normal state for algebras of type II or III \cite[Lemma 6.10]{ohyapetz}\cite[Lemma 2.4]{narnhofer94}.
Our work will exploit the split property and depend only on entropy of algebras of type I.


\subsection{Nuclear maps}\label{subsec:lp}

In order to discuss the conformal nuclearity condition, we first recall the notions of nuclear maps, $p$-nuclear maps,
and their nuclearity indices, together with some of their basic properties.
For this Section, we shall take \cite{bdl90a,bdl90b,fop05} as reference.

Our interest is only in the case $p \le 1$, where the discussion is much simpler than the general case.
For a generalization to the case $1 < p \le \infty$, see \cite{fop05}.

\begin{defi}\label{def:pnuc} Let $X,Y$ be Banach spaces and $p$ a parameter with $0<p\le 1$.
  A bounded linear operator $T: X\to Y$ is said to be {\bf $p$-nuclear} if there are families
  $ \{ \phi_k \}_{k\in\mathbb{N}} \subset  X^*$ of linear functionals on $X$ and 
  $ \{ \xi_k  \}_{k\in\mathbb{N}} \subset Y$ of vectors in $Y$ such that the following decomposition holds:
  \[
   T(\cdot) = \sum_{k\in\mathbb{N}} \phi_k(\cdot) \xi_k , \qquad 
   \textrm{ and } \qquad \sum_{k\in\mathbb{N}} \left(\|\phi_k\| \cdot \| \xi_k \|\right)^p \le +\infty.
  \]
  Furthermore, any such decomposition is called a {\bf $p$-nuclear decomposition},
  and we define $\nu_p(T)$ as the {\bf $p$-nuclearity index of $T$} given by
  \[\nu_p(T) := \inf \sum_{k=0}^\infty (\|\phi_k\|\cdot \| \xi_k\|)^p, \]
  with the infimum taken over all the $p$-nuclear decompositions as above.
  \end{defi}

\begin{prop}\label{pro:lpspace} For operators in $\B(X,Y)$, the following hold.
  \begin{itemize}
  \item Consider $0<p\le 1$. Then, $\nu_p(\cdot)$ is $p$-homogeneous and subadditive, namely if $T_1,T_2$ are two operators and $\lambda$ is a scalar, then
    \[\nu_p(\lambda T_1) = |\lambda|^p \cdot \nu_p(T_1), \quad \textrm{and} \quad
    \nu_p(T_1+T_2) \le \nu_p(T_1) + \nu_p(T_2). \]
  \item Consider $0<p\le 1$. Then the following holds:
    \[\nu_p( R S T) \le \|R\|^p \cdot \nu_p(S) \cdot  \|T\|^p \qquad (R\in\B(X),S\in\B(X,Y),T\in\B(Y)). \]
  \item Consider $0<p\le 1$. Then, $\nu_p(\cdot)^{(1/p)}$ is a quasi-norm, namely, all the axioms for norm holds except for the triangle inequality, which is replaced by the following, given any family $\{T_k \, , \, k=1,2,\ldots,N\}$ of operators: \[ \sum_{k=1}^N\nu_p(T_k)^{(1/p)} \le \nu_p\left(\,\sum_{k=1}^N T_k\,\right)^{(1/p)} \le N^{\frac{1-p}{p}}\sum_{k=1}^N\nu_p(T_k)^{(1/p)} .\]
  \item Consider $0<p\le q \le 1$. Then, $p$-nuclearity implies $q$-nuclearity.
  \end{itemize}
\end{prop}

\subsection{Nuclearity conditions and the split property}\label{subsec:nuc}

In this section, we discuss the conformal nuclearity condition and the split property, as well as other related conditions
which we will exploit.
Let us briefly recall its physical motivations for these conditions. 
The first of them was the Haag-Swieca compactness criterion \cite{haaswi65},
a criterion to exclude QFT with ``too many states'' such as generalized free fields.
It requires that the states generated by local observables from the vacuum, multiplied by a spectral projection of the Hamiltonian,
form a compact set.
The Buchholz-Wichmann energy nuclearity condition \cite{bucwic86} is a strengthened version of the above.
They replace the sharp cutoff by a smooth damping, requiring nuclearity instead of compactness,
namely, the set $e^{-\beta H}\mathcal{L}_\mathcal{O}$ is nuclear for all $\beta > 0$,
where $H$ is the Hamiltonian and $\mathcal{L}_\mathcal{O}$ is the vectors generated from the vacuum by
observables localized in $O$ with norm $1$. One can equivalently formulate this condition
as the nuclearity of the map from the local algebra $\A(O)$ into $\H$, which we also adopt in this work.

For chiral nets, we can consider a variation of the nuclearity conditions using $L_0$ instead of $H$,
and adopt the nomenclature ``conformal nuclearity condition'', which was used in \cite{bdl07}. Our precise definition follows.

\begin{defi}
\label{def:nuc} Let  $(\A,U,\Omega,\mathcal{H})$ be a M\"obius covariant local net.
For $I\in\mathcal{I}$ and $\beta>0$, define the damping map $\Theta_{I,\beta}: \A(I) \to \mathcal{H}$ by the following formula:
  \[
   \Theta_{I,\beta}: x\in\A(I) \mapsto e^{-\beta L_0}x\Omega \in \mathcal{H}.   
  \]
  With those maps, we define the conformal nuclearity conditions:
  \begin{enumerate}[(1)]
  \item \label{it:nucp} Let $p \in (0,1]$.
  The net satisfies the {\bf conformal $p$-nuclearity condition} if
  the map $\Theta_{I,\beta}$ is $p$-nuclear for any $I \in \mathcal{I}$ and $\nu_p(\Theta_{I,\beta}) \le \exp\left( \, (c_{I,p}/\beta)^{n_{I,p}} \right)$,
  where $c_{I,p}$ and $n_{I,p}$ are positive constants depending on $I$ and $p$.
  \item \label{it:nuc1} The net satisfies the {\bf conformal nuclearity condition} if the above holds for $p=1$.
\end{enumerate}\end{defi}
Note that we require an estimate on the nuclearity index, additional to
the ``conformal nuclearity condition'' of \cite[Section 6]{bdl07}.


We now address one of the consequences of the conformal nuclearity condition, the {\it split property}.
It is an algebraic property that relates to the statistical independence of two separated local algebras.

\begin{defi}
\label{def:split}
  Let $(\A,U,\Omega,\mathcal{H})$ be a M\"obius covariant local net on a separable Hilbert space $\H$.
  The net satisfies the {\bf split property} if,
  for any any $I_1, I_2 \in \mathcal{I}$ such that $\overline{I_1} \cap \overline{I_2} = \emptyset$ ({\it i.e.} $I_1\Subset I_2'$), the following equivalent properties hold:
  \begin{itemize}
  \item the (algebraic) $*$-homomorphism $a\otimes b\in \A(I_1)\otimes_{\rm alg}\A(I_2) \mapsto a\cdot b \in \A(I_1)\vee \A(I_2)$
  extends to an $*$-isomorphism of von Neumann algebras $\A(I_1)\otimes\A(I_2) \cong \A(I_1)\vee \A(I_2)$.
  \item the inclusion $\A(I_1) \subset \A(I_2)'$ is a standard split inclusion of von Neumann algebras (with respect to $\Omega$),
  i.e.\! $\Omega$ is a cyclic vector for $\A(I_1)$, $\A(I_2)'$ and $\A(I_1)'\cap\A(I_2)'$,
  and there is a von Neumann algebra $\R$ which is an intermediate factor type I,
  namely, $\A(I_1) \subset \R \subset \A(I_2)'$.
  \end{itemize}
\end{defi}
These definitions are indeed equivalent because the underlying Hilbert space is separable
and local algebras are type III factors, hence they are isomorphic if and only if they are unitarily equivalent
(spatially isomorphic). This will be used in the following without remark.
Actually, the second of the statements above forces the separability of the underlying Hilbert space \cite[Proposition 1.6]{doplon84}.

The intermediate factors of type I will be essential in our later definitions involving entropy.
Since there are many choices of such factors\footnote{There is a canonical choice of such a type I factor \cite{doplon84},
yet its physical meaning is not very clear and we consider all such intermediate type I factors.}, we shall adopt the following notation:

\begin{defi}\label{def:ursplit}
  For $(\A,U,\Omega,\mathcal{H})$ a M\"obius covariant local net satisfying the split property, and $I_1, I_2 \in \mathcal{I}$
  such that $I_1 \Subset I_2$, we use the symbol $(u,\R_u)$ to denote a pair of a unitary operator
  $u:\mathcal{H}\to\mathcal{H}\otimes\mathcal{H}$ implementing the $*$-isomorphism $\A(I_1)\vee\A(I_2) \cong \A(I_1)\otimes\A(I_2)$,
  and an intermediate type I factor $\R_u = u^*(\B(\H)\otimes \CC\1)u$.
\end{defi}

Having stated the nuclearity conditions and the split property,
we now make a list of useful additional assumptions for M\"obius covariant local nets.

\begin{defi}\label{def:ea} For a M\"obius covariant local net, we can consider the following additional conditions:
  \begin{enumerate}
  \item \label{it:da} $\dim \ker (L_0-N)\le C\exp(N^\kappa)$ for constants $\kappa \in (0,1)$ and $C > 0$.
  \item \label{it:tr} Trace class condition: there are positive parameters $a,b,c$ such that \[ \Tr(e^{-\beta L_0}) \le a\exp\left(b\beta^{-c}\right) \qquad (\textrm{for } \, \beta > 0).\]
  \item \label{it:pn} Conformal $p$-nuclearity condition for all $0 < p \le 1$ (Definition \ref{def:nuc}(\ref{it:nucp})).
  \item \label{it:1n} Conformal nuclearity condition (Definition \ref{def:nuc}(\ref{it:nuc1})).
  \item \label{it:sp} Split property (Definition \ref{def:split}).
  \end{enumerate}
\end{defi}

\begin{prop}\label{pro:ea} For a M\"obius covariant local net, there is a chain of implications between
the conditions in Definition \ref{def:ea}: 
(\ref{it:da}) $\Rightarrow$ (\ref{it:tr}) $\Rightarrow$ (\ref{it:pn}) $\Rightarrow$ (\ref{it:1n}) $\Rightarrow$ (\ref{it:sp}).
\end{prop}

\begin{proof}
  We first work on the implication (\ref{it:da}) $\Rightarrow$ (\ref{it:tr}). Consider the parameters $C>0$ and $\kappa\in(0,1)$ such that $\dim \ker (L_0-N)\le C\exp(N^\kappa)$. Then,
  \[ \Tr(e^{-\beta L_0}) = \sum_{N\ge 0} \dim \ker (L_0-N)e^{-\beta N}
  \le \sum_{N\ge 0} C e^{-\beta N + N^\kappa}. \]

  Since $-\beta N$ eventually dominates $N^\kappa$, the trace is always finite. All that is left is to verify the dependence on $\beta$. Our strategy is to divide the sum in three parts, the first term, a finite sum, and a infinite sum with exponential decrease.

  The exponent can be expressed as $-\beta N + N^\kappa = -\beta N^\kappa (N^{1-\kappa}-1/\beta)$.
  Note that there is a number $A$ such that $N^{1-\kappa}-1/\beta \ge A^{1-\kappa}-1/\beta  > 0$ whenever $N\ge A$. Indeed, one can take any $A$ such that $A > \beta^{-1/(1-k)}$, but we shall fix this value later. Defining $B := A^{1-\kappa}-1/\beta > 0$, one has that $-\beta N+N^\kappa \ge -B\beta N^\kappa$ for $N\ge A$. Hence, dividing the sum in $\{N=0\}$, $\{0<N\le A\}$ and $\{N > A\}$, the first and the later can be bound by the following inequality:  
  \begin{align*} 1 + \sum_{N>A} e^{-\beta N + N^\kappa} &=
  1 + \sum_{N> A} e^{-\beta N^\kappa(N^{1-\kappa}-1/\beta)} \\ &
  \le 1 + \sum_{N> A} e^{-B \beta N^\kappa} \le \sum_{N\in \mathbb{N}} e^{-B \beta N^\kappa}. \end{align*}
  Now, to turn the last term above in a quantity independent of $\beta$, we pick $A$ as following: \[ A := (2/\beta)^{1/(1-\kappa)}, \qquad B=1/\beta. \]
  The sum $\{N=0\}\cup\{N> A\}$ is then bounded by a constant expressed in the following:
  \[1+\sum_{N> A} e^{-\beta N + N^\kappa} \le \sum_{N\in \mathbb{N}} e^{-N^\kappa}.\]
  
  The remaining finite sum can be bounded by the number of terms times the supremum of the function. We first analyze the exponent $-\beta N + N^\kappa$. By elementary calculus, it takes its maximal value at $N_0 = (\kappa/\beta)^{1/(1-\kappa)}$, and hence one has the supremum bound
  \[ \sup_{N\ge 0}\big|e^{-\beta N + N^\kappa}\big| = \exp\Big(\Big((1-\kappa)\,\kappa^{-\frac{\kappa}{1-\kappa}}\Big) \beta^{-\frac{\kappa}{1-\kappa}} \Big). \]
Moreover, the number of terms is $\#\{0<N\le A\} = \lfloor A \rfloor \le A = (2/\beta)^{1/(1-\kappa)}$. One has the following bound for the finite part of the sum:
  \begin{align*} C\sum_{0<N\le A} e^{-\beta N + N^\kappa} \le CA \|e^{-\beta N + N^\kappa}\|_\infty
    & \le \frac{C2^{1/(1-\kappa)}}{\beta^{1/(1-\kappa)}} \exp\Big(\Big((1-\kappa)\,\kappa^{-\frac{\kappa}{1-\kappa}}\Big) \beta^{-\frac{\kappa}{1-\kappa}} \Big)\\
    & = \frac{a_0}{\beta^{c_0}} \exp\Big(b_1 \beta^{-c_1}\Big) \\
    & \le a_2 \exp\Big(\beta^{-c_2}\Big) \qquad (\beta>0).
  \end{align*}
  where the constants $a_0,c_0,b_1,c_1$ are easily identifiable.
  To identify $a_2,c_2$, we notice that $\beta^{-c_0} \le \exp(\beta^{-c_0})$, and put $c_2 = \max\{c_0,c_1\}$.
  Then, $a_2$ is such that $\beta^{-c_0} + b_1\beta^{-c_1} \le \log a + \beta^{-c_2}$ holds for all $\beta>0$
  (e.g.\! $\log a_2 = 1 + b_1^{(c_2-c_1+1)/(c_2-c_1)}$).

  Therefore, the trace satisfies the following inequality.
  \[ \Tr(e^{-\beta L_0}) \le C\sum_{N\in \mathbb{N}} e^{N^\kappa} + a_2\exp\big(\beta^{-c_2}\big) \le a\exp\big(b\beta^{-c}\big) \qquad (\beta>0), \]
  where $a = C\sum_{N\in \mathbb{N}} e^{N^\kappa} + a_2$, $b=1$, and $c=c_2$. 
  This concludes the proof of (\ref{it:da}) $\Rightarrow$ (\ref{it:tr}).

  The implication (\ref{it:tr}) $\Rightarrow$ (\ref{it:pn}) follows since, for any $I\in\mathcal{I}$, $0<p\le 1$ and $\beta >0$, the inequality $\nu_p(\Theta_{I,\beta}) \le \Tr(e^{-p\beta L_0}) \le a\exp\left( (b/p^c)\,\beta^{-c}\right)$ holds. 
  The implication (\ref{it:pn}) $\Rightarrow$ (\ref{it:1n}) is trivial. Finally, we refer the proof of (\ref{it:1n}) $\Rightarrow$ (\ref{it:sp}) to references, see \cite[Lemma 2.12]{gabfro93}, which translates the arguments of \cite[Section 2]{bdf87} to the chiral setting. See also \cite[Corollary 6.4]{bdl07} for a different proof that holds also in a ``distal'' case, involving concepts of modular nuclearity and $L^2$-nuclearity.
\end{proof}

 The $p$-nuclearity condition was introduced in \cite{bucpor90}, in the investigations of the phase space in AQFT.
 Later, in \cite{fop05}, the formulation is further clarified and meaningfully defined for $p>1$.
 For chiral nets, the conformal nuclearity condition was studied in \cite{bdl07}, together with other notions of nuclearity.
 Also, in \cite{mtw16} it was proved that the split property follows automatically if the chiral net is $\Diff$-covariant.


To conclude this section, let us remark that
most of known M\"obius covariant local nets with nice properties satisfy the condition (\ref{it:da}).
For example, as for the $\mathrm{U}(1)$-current net,
the dimension of $L_0$-eigenspaces grows as the partition function $p(N)$, which behaves asymptotically as
$p(N) \sim \frac{1}{4\sqrt{3}N}e^{\pi\sqrt{\frac23}N^\frac12}$ \cite[24.2.1.III]{abraste}.
The same applies to the Virasoro nets with $c > 1$ \cite{KR87}.
Some completely rational nets can be realized as a subnet of a tensor product of copies of the free fermion net,
which is nicely summarized in \cite[Section 4.2]{Tener16}, and a similar estimate can be made there.


\section{Towards cutoff entropy for a chiral net}\label{sec:cut}

Throughout this Section, let $(\A,U,\Omega,\H)$ be a M\"obius
covariant local net satisfying the split property (Definition
\ref{def:split}). Let $\omega = \inner{\Omega}{\cdot\,\Omega}$ be the
vacuum state, and $L_0$ the conformal Hamiltonian. We also fix an
interval $I\in\mathcal{I}$.

As mentioned before, when trying to define the entanglement entropy of
the vacuum with respect to $I$, that is, the quantum entropy of
$\omega$ as a state restricted to $\A(I)$, one has to deal with the
fact that $\A(I)$ is a von Neumann algebra of type III, whose natural variation of von Neumann entropy
is divergent \cite[Lemma 6.10]{ohyapetz}\cite[Lemma 2.4]{narnhofer94}. Furthermore, one has to deal with
UV-divergences expected from the physical literature. One therefore needs
a certain regularization.

We define our entropic quantities with various regularizations in three steps.
\begin{enumerate}[(1)]
\item Let $\delta>0$ be a parameter such that $|I| + 2\delta < 2\pi$.
We put $I_\delta = \cup_{-\delta < \theta <\delta}\;\rho_\theta(I) \in \mathcal{I}$ the ``augmentation of $I$ by $\delta$''.
It then holds that $I \Subset I_\delta$, and the split property asserts that there are pairs $(u,\R_u)$, where $\R_u$ is a type I factor
such that $\A(I) \subset \R_u \subset \A(I_\delta)$. We introduce the quantity $H_{I,\delta}(\omega)$ with the aid of intermediate type I factors $\R_u$ (note that, however, this will \emph{not} be $S_{\R_u}$ of Definition \ref{def:SvnI}).
\item We can estimate $H_{I,\delta}(\omega)$ from above, provided the net satisfies condition \ref{def:ea}(\ref{it:da}). However, this estimate diverges as $\delta$ approaches zero.
We then regularize the states by a cutoff parameter $E$, and define the regularized quantity $H^E_{I,\delta}(\omega)$.
\item Finally, we consider $H_I^E(\omega)$ with cutoff $E$, as the limit of the former as $\delta$ goes to zero.
We also state our main result, Theorem \ref{thm:cut3}, stating the finiteness of it, with an upper bound given in terms of the dimensions of eigenspaces of the conformal Hamiltonian.
The proof of it is spread in the later sections.
\end{enumerate}


\subsection{The energy function}

Before further discussion, we recall an important lemma from
\cite{bdf87} which states the existence of an auxiliary energy
function $f$ that will be necessary for our calculations.  While in
the original paper $f$ is an almost exponentially decreasing function,
namely it holds for any $\kappa \in (0,1)$ that $\sup_{t\in\mathbb{R}}
\left| f(t)\right| \exp(|t|^\kappa)< \infty$, it turns out that for
our needs the choice of an energy function becomes more flexible. We
reproduce the proof in order to stress this point.

\begin{lemm}\label{lem:f} (cf. \cite[Lemma 2.3]{bdf87})
Let $(\A,U,\Omega,\mathcal{H})$ be a M\"obius covariant local net. 
Let $\alpha$ be a parameter such that $0<\alpha<1$.
Then there is an energy function $f : t\in\mathbb{R} \mapsto f(t) \in \mathbb{R}$
with the following properties:
\begin{enumerate}
\item The function $f$ satisfies $\sup_{t\in\RR} |f(t)| \exp(|t|^{\alpha}) < \infty$ and $f(0)=1/2$.

\item If $x,y$ are local operators such that $[x, \Ad e^{i\theta L_0}(y)] = 0$ whenever $|\theta|<1$,
it holds that
\[ \inner{\Omega}{xy\,\Omega} = \binner{\Omega}{\left(x\,f(L_0)\,y + y\,f(L_0)\,x\right)\, \Omega} .\]

\item Let $f_\delta$ be the $\delta$-scaled $f$, that is, $f_\delta(t) := f(\delta t)$. Then, for any pair of local operators $x,y$
such that $[x, \Ad e^{iL_0t}(y)] = 0$ holds whenever $|t|<\delta$ ({\it i.e.} the distance of their respective local algebras
is larger than $\delta$), it holds that
\[ \inner{\Omega}{xy\,\Omega} = \binner{\Omega}{\left(x\,f_\delta(L_0)\,y + y\,f_\delta(L_0)\,x\right)\, \Omega} .\]

\end{enumerate}\end{lemm}

\begin{proof}
{\bf (1)(2): Existence of the function.}
The argument follows the proof of Lemma 2.3 in \cite{bdf87},
with the conformal Hamiltonian $L_0$ taking place instead of the Hamiltonian $H$.
Consider any two local operators $x,y$ satisfying $[x, e^{i\theta L_0}ye^{-i\theta L_0}] = 0$ whenever $\theta \in (-1,1)$. This implies that
\begin{equation}\label{eq:koko}
\inner{\Omega}{x e^{itL_0} y \Omega} = \inner{\Omega}{y e^{-itL_0} x \Omega}, \quad  t\in (-1,1).
\end{equation}
By the positivity of $L_0$, the left-hand side extends continuously to $\{\zeta \in \mathbb{C}, \im \zeta \ge 0\}$
and analytically to its interior, and likewise,
the right-hand side extends to $\{\zeta \in \mathbb{C}, \im \zeta \le 0 \}$.
Therefore, there is a holomorphic function $h$ defined on $\mathcal{P}_1 = \mathbb{C}\setminus\left( (-\infty,-1]\cup[1,+\infty)\right)$
such that on $(-1,1)$ it coincides with the function expressed in \eqref{eq:koko}.

Next, fix a constant $\tau \in (0,1)$, and consider the conformal map that takes the disc $\mathbf{D} = \{ w\in\mathbb{C}, |w|<1 \}$
onto $\mathcal{P}_\tau := \mathbb{C}\setminus\left( (-\infty,-\tau]\cup[\tau,+\infty)\right)$, given by $z_\tau(w) = 2\tau w/(w^2 +1)$,

As $\mathcal{P}_\tau \subset \mathcal{P}_1$ for $\tau \in (0,1)$,
the function $h_\tau(w) := h(z_\tau(w))$ is holomorphic on $\mathbf{D}$,
and it is easy to see that $h_\tau$ is continuous and bounded (by $\|x\|\cdot\|y\|$)
on $\overline{\mathbf{D}}\setminus \{\pm 1 \}$. Therefore, by integrating $\frac1 w h_\tau(w)$ on a circular path $w(s) = re^{is} \in \mathbf{D}$
with a fixed radius $r<1$ and parameter $s \in (0,2\pi)$, one can invoke Cauchy's residue theorem and take $r \nearrow 1$ to get the following equality for all $\tau \in (0,1)$:
\begin{equation}\label{eq:pwpw}
\inner{\Omega}{xy\Omega} = \frac{1}{2\pi}\int_{0}^{\pi}ds\,\binner{\Omega}{(xe^{iL_0\tau/\cos(s)}y + ye^{iL_0\tau/\cos(s)}x) \Omega}
\end{equation}

For $\alpha$ as in the statement, fix a value $\beta$ such that $\alpha < \beta < 1$.
There exists a smooth function $g$ such that $\tilde{g}$ is smooth and supported inside $(0,1)$ such that
$g$ decays as $e^{-|t|^{\beta}}$ for $|t|$ large with $g(0) = 1$
(see \cite{jaffe67} and the references therein
for the existence functions of almost exponential decay with compact Fourier transform,
and \cite{johnson15} for more concrete functions with weaker requirements needed here).
One can then multiply the above equality by $\tilde{g}(\tau)$ and integrate it against $d\tau$ to then obtain
\begin{equation*}
\inner{\Omega}{xy\Omega} = \binner{\Omega}{(xf(L_0)y + yf(L_0)x) \Omega}, 
\end{equation*}
where $f$ is defined by $f(t) = (2\pi)^{-1} \int_{0}^{\pi} g\left(t/\cos(s)\right)\,ds$.
Putting $t = 0$ shows that $f(0) = 1/2$, as we took $g(0) = 1$. Furthermore, $f$ inherits the decay property from $g$.

Until now, $f$ is a complex-valued function.
Yet, it is immediate that the function $\bar f(t) = \overline{f(t)}$
has the same property:
\begin{align*}
\binner{\Omega}{(x\bar f(L_0)y + y\bar f(L_0)x) \Omega} &= \binner{(y^*f(L_0)x^* + x^*f(L_0)y^*) \Omega}{\Omega} \\
&= \overline{\binner{\Omega}{(y^*f(L_0)x^* + x^*f(L_0)y^*) \Omega}} \\
&= \overline{\binner{\Omega}{(x^*y^*) \Omega}} \\
&= \binner{\Omega}{xy \Omega}
\end{align*}
Therefore, the real part of $f$ does the same job.
In the following, we assume that $f$ is real.

{\bf (3): Scaling.} Now, considering the parameter $\delta$, consider two local operators $a,b$
satisfying the commutation rule $[x,e^{i\theta L_0t}ye^{-i\theta L_0}] = 0$ whenever $\theta \in (-\delta,\delta)$.
The previous discussion follows analogously, except that the equality given by
equation \eqref{eq:pwpw} holds only for $\ \in (-\delta,\delta)$.
In following, using $\tilde{g}(\tau/\delta)$  instead of $\tilde{g}$ we obtain the equality
\begin{equation*}
\inner{\Omega}{xy\Omega} = \binner{\Omega}{(xf_\delta(L_0)y + yf_\delta(L_0)x) \Omega}, 
\end{equation*}
where now $f_\delta(t) = f(\delta t)$, thus proving (3).
\end{proof}


\subsection{Regularization by distance}\label{subsec:cutdef}


First, regarding the von Neumann entropy (as introduced in Section \ref{subsec:Svn}), we make an observation which will stand as motivation of our definition.

\begin{remark} Given two normal states $\varphi,\psi$ on $\B(\H)$, we say $\varphi\succeq\psi$ if there is a positive number $t>0$ such that $t\varphi \ge \psi$, and equivalently, if there is a positive number $\lambda \in (0,1]$ such that $\varphi \ge \lambda \psi$ (here, $t = 1/\lambda$). The concavity of the von Neumann entropy asserts that $\Svn(\varphi) \ge \lambda \Svn(\psi)$. We therefore have
\[ \Svn(\psi) = \inf_{\varphi} \frac{1}{\lambda_\varphi}\Svn(\varphi), \]
where the infimum runs over all states $\varphi$ to which there is a positive parameter $\lambda_\varphi \in (0,1]$ such that $\varphi \ge \lambda_\varphi \psi$. Clearly, equality holds since $\psi \succeq \psi$ with $\lambda_\psi = 1$.
\end{remark}


Turning back to conformal nets, we recall the split property (Definition \ref{def:split}).
Consider an interval $I\in \mathcal{I}$ and a positive parameter $\delta>0$.
Let $I_\delta = \cup_{-\delta < \theta <\delta}\;\rho_\theta(I)$
be the the ``augmentation of $I$ by $\delta$'', where $|I| + 2\delta < 2\pi$
so that it holds that $I_\delta \in \mathcal{I}, I \Subset I_\delta$.
By the split property, there are pairs $(u,\R_u)$ as in Definition \ref{def:ursplit},
where $u:\H\to\H\otimes\H$ is unitary such that $u(xy)u^* = x\otimes y$ for any pair $(x,y) \in \A(I) \times \A(I_\delta)'$,
and $\R_u = u^*(\B(\H)\otimes \CC\1)u$ is an intermediate type I factor.

Since the entanglement entropy is can be defined through the formula of von Neumann for type I factors $\R_u$
(denoted accordingly as $S_{\R_u}$, see Definition \ref{def:SvnI}),
one can consider the following.

\begin{defi}\label{def:cut1} Consider a normal state $\psi$ on $\B(\H)$.
For $I\in\mathcal{I}$ and $\delta>0$ (such that $I_\delta\in\mathcal{I}$), we define
  \[ H_{I,\delta}(\psi) := \inf_{(u,\R_u)}\;\inf_{\varphi} \; \frac{1}{\lambda_\varphi}S_{\R_u}(\varphi), \]
 The first infimum runs over all pairs $(u,\R_u)$ as in Definition \ref{def:ursplit},
 and the second infimum runs over all normal states $\varphi$ over $\B(\H)$ to which there is a positive number $\lambda_\varphi\in(0,1]$ such that $\varphi  \ge \lambda_\varphi \psi$ when restricted to $\A(I)\vee\A(I_\delta)'$.
\end{defi}

\begin{remark}
 Let us consider type I situations. Let $\B(\H) = \B(\H_1)\otimes \B(\H_2)\otimes \B(H_3)\otimes \B(\H_4)$,
 and $\psi$ be a pure state on $\B(\H_1)\otimes \B(\H_2)$.
 We take a state $\tilde \psi$ on $\B(\H)$ which extends $\psi$. As $\psi$ is pure,
 it must be of the form $\tilde \psi = \psi \otimes \varphi$.
 It is easy to see that $S_{1,3}(\tilde\psi) = S_1(\psi) + S_3(\varphi)$,
 and $S_3(\varphi)$ can be zero if we take a tensor product $\varphi = \varphi_3\otimes \varphi_4$ of pure states.
 Hence for a pure state $\psi$ we have $\Svn(\psi) = \inf_{\tilde\psi} \Svn(\tilde \psi)$
 where $\inf$ runs over all extended state to a type I factor $\B(\H)$ ($\H$ is not necessarily fixed).
 We believe that this justifies our Definition \ref{def:cut1}.
 
 Furthermore, let us point out that there are several possible definitions of entanglement entropy
 which coincide with the von Neumann entanglement entropy when the state is pure, see e.g.\! \cite[Theorem 3]{VP98}.
 As our main purpose is type III algebras which admit no normal pure state, we have to make a choice.
 See \cite{HS17} for a different choice, analogous to that of \cite[Section D.1]{VP98}
\end{remark}

The definition \ref{def:cut1} relies on pairs $(u,\R_u)$ and the von Neumann entropy of states $\varphi$ restricted to $\R_u$.
We do the actual calculations through the unitary $u$, considering $\varphi\comp\Ad_{u^*}$ as a state on $\H\otimes\H$
and restricting it to the first tensor component. We state this fact in the following lemma.

\begin{lemm}\label{lem:hh} Consider $\psi$ a normal state in $\B(\H)$. For $\delta>0$ fixed, let $(u,\R_u)$ be as in Definition \ref{def:ursplit}. Let $\varphi$ be a normal positive functional on $\B(\H\otimes\H)$ such that $\varphi\comp \Ad_{u^*} \ge \psi$ on $\A(I)\vee\A(I_\delta)'$. Then,
  \[ \varphi(\1) \cdot S_1\left(\frac{\varphi}{\varphi(\1)}\right) \ge H_{I,\delta}(\psi)  ,\]
  where  $S_1$ is the von Neumann entropy of a state in $\B(\H) \otimes \B(\H)$ restricted to the first tensor component.
  \end{lemm}

\begin{proof}
  Call $\tilde{\varphi} = \varphi\comp\Ad_{u^*}$ and $\varphi_1(\cdot)=\varphi(\,\cdot\otimes \1)$.
  From definitions \ref{def:cut1}, we have $H_{I,\delta}(\psi) \le \tilde{\varphi}(\1) \, S_{\R_u}(\tilde{\varphi}/\tilde{\varphi}(\1))$.
  It suffices to show that $S_{\R_u}(\tilde{\varphi}/\tilde{\varphi}(\1)) = S_1(\varphi/\varphi(\1)) = \Svn(\varphi_1/\varphi_1(\1))$, and by Definition \ref{def:SvnI}, it suffices to show that there is an $*$-isomorphism $\s: \R_u \to \B(\H)$ such that $\tilde{\varphi} = \varphi_1\comp \s$. Then, $\s$ defined by $\s^{-1}: x\in\B(\H) \mapsto u^*(x\otimes \1)u \in \R_u$ satisfies the requirement.
\end{proof}

Our main objective in this section is to prove the following.

\begin{prop}\label{pro:cut1} Let be a M\"obius covariant local net satisfying the condition \ref{def:ea}(\ref{it:da}), \textit{i.e.}\! there are constants $\kappa \in (0,1)$ and $C>0$ such that $\dim\ker(L_0-N)\le Ce^{N^\kappa}$.
Then, for the vacuum state $\omega$ restricted to $I\in\mathcal{I}$ with regularization parameter $\delta$ (such that $I_\delta\in\mathcal{I}$,
the quantity $H_{I,\delta}(\omega)$ is finite. More precisely,
\[
 H_{I,\delta}(\omega) \le C_\delta\log C_\delta + S_\delta,
\]
where $C_\delta$ and $S_\delta$ are given by
\begin{align*}
  C_\delta &= \sum_{N\ge 0} 2 \dim(\H_N) |f_\delta(N)|,\\
  S_\delta &= \sum_{N > 0} 4 \dim(\H_N) \left(-\frac{\left|f_\delta(N)\right|}2\log\frac{\left|f_\delta(N)\right|}2\right),
\end{align*}
with $f$ the energy function as in Lemma \ref{lem:f} with some $\alpha$ satisfying $0<\kappa<\alpha<1$.
\end{prop}

The split property guarantees the existence of a unitary operator $u_\delta: \H\to\H\otimes\H$ intertwining $\A(I)\vee\A(I_\delta)'$ and $\A(I)\otimes\A(I_\delta)'$,
and it follows from the assumption \ref{def:ea}(\ref{it:da}). Therefore, the quantity $H_{I,\delta}(\omega)$ can be defined as above.

\begin{proof}
We fix a parameter $\alpha$ such that $0<\kappa<\alpha<1$ and invoke Lemma \ref{lem:f} to obtain an energy function $f:\mathbb{R}\to\mathbb{R}$
such that $\sup_{t\in\mathbb{R}}\left|e^{|t|^\alpha}f(t)\right| < +\infty$. Recall that the following holds:
\[
\omega(xy) = \omega(x f_\delta(L_0) y + y f_\delta(L_0) x), \qquad \text{ for } x\in\A(I), y\in\A(I_\delta)'
\]

Define $\theta_\delta$ as the self-adjoint linear functional on $\B(\H\otimes\H)$ given by the following formula:
\begin{equation}\label{eq:thetad}
 \theta_\delta(x\otimes y) := \omega( x f_\delta(L_0) y + y f_\delta(L_0) x).
\end{equation}
This indeed defines a normal linear functional because of the decay property of $f$
and is self-adjoint because $f$ is real.

The vacuum state $\omega$ and the functional $\theta_\delta \comp \Ad_{u_\delta^*}$  coincide in $\A(I)\vee{\otimes}\A(I_\delta)'$, so one might be tempted to invoke Lemma \ref{lem:hh} and state that $H_{I,\delta}(\omega) \le S_1(\theta_\delta)$. However, one should note that $\theta_\delta$ is only positive when restricted to the above-mentioned algebra. Nevertheless, one can decompose it as $\theta_\delta = \theta_{\delta,+} - \theta_{\delta,-}$, where $\theta_{\delta,\pm}$ are positive functionals on $\B(\H\otimes\H)$
(this will {\it not} be the Jordan decomposition, the detailed construction of $\theta_{\delta,\pm}$ will be given below, in particular on equation \eqref{eq:thetapm}).
Hence, restricted to $\A(I)\otimes\A(I_\delta)'$, one has $\theta_{\delta,+} = \theta_\delta + \theta_{\delta,-} = \omega\comp\Ad_{u_\delta^*} + \theta_{\delta,-}$, and therefore, after normalizing it to $\hat{\theta}_{\delta,\pm} = \theta_{\delta,\pm}/\|\theta_{\delta,\pm}\|$, one has that $\hat{\theta}_{\delta,+}\comp\Ad_{u_\delta} \ge (1/\|\theta_{\delta,+}\|) \omega$. By Lemma \ref{lem:hh}, one has the inequality
\begin{equation*}
H_{I,\delta}(\omega) \le \|\theta_{\delta,+}\|S_{\R_u}(\hat{\theta}_{\delta,+}\comp \Ad_{u^*}) = \|\theta_{\delta,+}\|S_1(\hat{\theta}_{\delta,+}). 
\end{equation*}

It suffices then to identify the positive functional $\theta_{\delta,+}$ and find and an upper bound for the entropy $S_1(\hat{\theta}_{\delta,+})$. This will be addressed on the following.

\paragraph{The auxiliary functional $\theta_{\delta,+}$ and its entropy $S_1(\theta_{\delta,+})$}

We first further analyze the properties of $\theta_\delta$ to appropriately define a decomposition $\theta_\delta = \theta_{\delta,+} - \theta_{\delta,-}$.
The conformal Hamiltonian $L_0$ has discrete eigenvalues $N\in\mathbb{N}$ with eigenspaces $\H_N=\ker(L_0-N)$ of finite dimension
$\dim(\H_N)$. Let $\{\Phi_n\}_n$ be a normalized basis which consists of the eigenvectors of $L_0$ with eigenvalues $l_n \in \mathbb{N}$. Then, following the definition of $\theta_\delta$ as in Equation \eqref{eq:thetad}, one has

\begin{equation}\label{eq:thetan}
  \theta_\delta(x\otimes y) = \sum_{n\in \NN} f_\delta(l_n)\left(
  \inner{\Omega}{x\Phi_n}\inner{\Phi_n}{y\Omega} +
  \inner{\Omega}{y\Phi_n}\inner{\Phi_n}{x\Omega} \right).
\end{equation}

We proceed by decomposing the terms given by $\inner{\Omega}{x\Phi_n}\inner{\Phi_n}{y\Omega} + \inner{\Omega}{y\Phi_n}\inner{\Phi_n}{x\Omega}$ as a linear combination of positive terms.
First, we note that the $n$-sum has a special value at $n=0$, which account for the state $\omega\otimes\omega$, with multiplicity one since $f_\delta(0)=1/2$. We therefore focus on the terms corresponding to $n>0$.

We introduce $\phi_{k,n}$ as pure states on $\H$, defined as the following for $k = 0,1,2,3$ and $n>0$:

\begin{equation*}
  \phi_{k,n}(\cdot) := \left\<\frac{(\Omega+i^k\Phi_n)}{\|\Omega+i^k\Phi_n\|},
   \cdot \; \frac{(\Omega+i^l\Phi_n)}{\|\Omega+i^l\Phi_n\|}\right\>
  = \frac{1}{2} \binner{(\Omega+i^k\Phi_n)}{\;\cdot\; (\Omega+i^k\Phi_n)},
\end{equation*}
where the second equality follows since $\Omega$ and $\Phi_n$ are orthogonal to each other.

Standard algebraic manipulations show that the following polarizations hold for $n>0$:
\[\inner{\Omega}{x\Phi_n} = \sum_{k=0}^3 \frac{i^{-k}}{2}\phi_{k,n}(x)
  \qquad {\rm and } \qquad
  \inner{\Phi_n}{y\Omega} = \sum_{k=0}^3 \frac{i^{k}}{2}\phi_{k,n}(y)\]

  For $n>0$, the terms $\inner{\Omega}{x\Phi_n}\inner{\Phi_n}{y\Omega}$ and
  $\inner{\Omega}{y\Phi_n}\inner{\Phi_n}{x\Omega}$
  appearing in $\theta_\delta$ can then be written as a linear sum of positive functionals as follows:
\begin{align*}
  \inner{\Omega}{x\Phi_n}\inner{\Phi_n}{y\Omega} =
  & \frac{1}{4}\sum_{k,m=0}^3 \phi_{k,n}(x) \cdot i^{(+m)} \phi_{k+m,n}(y),\\
  \inner{\Omega}{y\Phi_n}\inner{\Phi_n}{x\Omega} =
  & \frac{1}{4}\sum_{k,m=0}^3 \phi_{k,n}(x) \cdot i^{(-m)} \phi_{k+m,n}(y).
\end{align*}

And hence:
 \[
  \inner{\Omega}{x\Phi_n}\inner{\Phi_n}{y\Omega} + \inner{\Omega}{y\Phi_n}\inner{\Phi_n}{x\Omega} =
  \frac{1}{2}\sum_{k=0}^3 \phi_{k,n}(x) \cdot \left(\phi_{k,n}(y) - \phi_{k+2,n}(y)\right).
 \]

In Equation \eqref{eq:thetan}, the above terms show up in $\theta_\delta$ multiplied by $f_\delta(l_n)$.
Aside from the value $l_n=0$ for which we know $f_\delta(0) = 1/2$,
each $f_\delta(l_n)$ might be positive or negative (as we noted, we can and do take a real $f$).

We then just need to be cautious about the sign of $f_\delta(l_n)$. Thus, we define:

\begin{equation*}
  a_\delta(k) := \left\{ \begin{array}{cl}
  1 & \textrm{ if } f_\delta(k) > 0 \\
  0 & \textrm{ otherwise}
\end{array}\right. \qquad
b_\delta(k) := \left\{ \begin{array}{cl}
  1 & \textrm{ if } f_\delta(k) < 0 \\
  0 & \textrm{ otherwise}
\end{array}\right. \end{equation*}

Then, for each $n$ at most one of the indices $a_\delta(l_n)$ and $b_\delta(l_n)$ is $1$,
and it holds that $f_\delta(k) = (a_\delta(k) - b_\delta(k))\cdot \left|f_\delta(k)\right|$.
Summing all terms,
we obtain
\begin{align*}
  \theta_\delta =& \; \underbrace{ \omega\otimes\omega +
  \sum_{n>0}\,\sum_{k=0}^3\, \frac{| f_\delta(l_n)|}{2}\,\phi_{k,n} \otimes \left( a_\delta(l_n) \phi_{k,n} + b_{\delta}(l_n)\phi_{k+2,n} \right)}_{=: \, \theta_{\delta,+}} \\
  & \; \phantom{\omega\otimes\omega} -\underbrace{
  \sum_{n>0}\,\sum_{k=0}^3\, \frac{| f_\delta(l_n)|}{2}\, \phi_{k,n} \otimes \left( a_\delta(l_n) \phi_{k+2,n} + b_{\delta}(l_n)\phi_{k,n}
  \right)}_{=: \, \theta_{\delta,-}} \\
\end{align*}
Hence we get the desired decomposition $\theta_\delta = \theta_{\delta, +} - \theta_{\delta, -}$ with  $\theta_{\delta,\pm}$ defined as
\begin{align}\label{eq:thetapm}
  \theta_{\delta,+} &:= \omega\otimes\omega +
  \sum_{n>0}\,\sum_{l=0}^3\, \frac{|f_\delta(l_n)|}{2}\,\phi_{l,n} \otimes \left( a_\delta(l_n) \phi_{l,n} + b_{\delta}(l_n)\phi_{l+2,n} \right). \nonumber\\
  \theta_{\delta,-} &:= \phantom{\omega\otimes\omega +}
  \sum_{n>0}\,\sum_{l=0}^3\, \frac{|f_\delta(l_n)|}{2}\,\phi_{l,n} \otimes \left( a_\delta(l_n) \phi_{l+2,n} + b_{\delta}(l_n)\phi_{l,n} \right).
\end{align}


With the definition of equation \eqref{eq:thetapm} in hands, we now focus on estimating the entropy $S_1(\theta_{\delta,+}/\|\theta_{\delta,+}\|)$.

Define $\tau_\delta$ as the positive functional $\theta_{\delta,+}$ restricted to the first tensor component. It is then expressed as follows:
\begin{equation}\label{eq:taud}
  \tau_\delta(x) := \theta_{\delta,+}(x\otimes \1) = \omega +
  \sum_{n>0}\sum_{k=0}^3 \frac{|f_\delta(l_n)|}{2}\,\phi_{k,n}(x)
  \qquad \text{ for } x\in \B(\H).
\end{equation}
This decomposition of $\tau_\delta$ into pure states $\phi_{k,n}$ is
indeed convergent in norm, because we assume that the net satisfies the condition \ref{def:ea}(\ref{it:da}), \textit{i.e.}\! $\dim(\H_N)$ grows bounded by an almost exponential function $C\exp(N^\kappa)$ for some $C>0$
and $\kappa\in(0,1)$.

The above decomposition into vector states enables us to invoke Corollary \ref{pro:Svnsum}.
We first note that, since the operator $\theta_{\delta,+}$ is not normalized, so is $\tau_\delta$ not normalized.
One can calculate its norm $C_\delta := \|\tau_\delta\| = \|\theta_{\delta,+}\| = \theta_{\delta,+}(\1\otimes \1)$ as
\[
C_\delta = \sum_{N\ge 0} 2 \dim(\H_N) |f_\delta(N)|,
\]
because every factor $|f_\delta(N)|/2$ appears $4\times\dim(\H_N)$, where the factor $4$ is due to the sum in $k$.
Likewise, with the fact that all $\phi_{k,n}$ are pure states, the same Corollary \ref{pro:Svnsum} tells us that the von Neumann entropy can be bounded as follows:
\[
\Svn\left(\frac{\tau}{\|\tau\|}\right) \le
\log C_\delta - \sum_{N > 0} \frac{4 \, \dim(\H_N)}{C_\delta} \cdot\left(\, \frac{|f_\delta(N)|}{2} \log \frac{|f_\delta(N)|}{2} \,\right),
\]
where the term $N = 0$ can be dropped in the second term because it corresponds to $\omega$ which is normalized.
This expression gives indeed a finite number, thanks to the condition \ref{def:ea}(\ref{it:da}) and the fact that $\kappa<\alpha$.
Hence, we have the upper bound for $H_{I,\delta}(\omega)$ as follows:
\[
  H_{I,\delta}(\omega) \le C_\delta\log C_\delta - \sum_{N > 0} 4 \, \dim(\H_N) \cdot\left(\, \frac{|f_\delta(N)|}{2} \log \frac{|f_\delta(N)|}{2} \,\right).
\]
\end{proof}

\subsection{Implementing cutoff}\label{subsec:cutoff}

We now consider a cutoff parameter $E$. We need this since the above quantity $H_{I,\delta}$ is expected to diverge when the spatial separation $\delta$,
taken as a variable parameter, approaches zero. 
We first define the regularization of states, and with those, we define the regularized entropy.

For any $E > 0$, let $P_E$ denote the spectral projection of the conformal Hamiltonian $L_0$
with respect to the set $[0, E]$. The set $\{P_E\}_{E > 0}$ is then an increasing family of projections
(acting on $\H$) indexed by a parameter $E > 0$, such that $P_E$ strongly converges to the unity as $E$ goes to infinity.

\begin{defi}\label{def:cutstate} Let $\phi$ be a normal positive functional on $\B(\H)$, and let $(u,\mathfrak{R}_u)$ be
as in Definition \ref{def:ursplit}. For $E\in \NN$, the regularized functional $\phi^{E,u}$ is defined as \[ \phi^{E,u} := x\in \B(\H) \mapsto \phi\left((u^*(P_E\otimes \1)u)\,x\,(u^*(P_E\otimes \1)u) \right) \in \mathbb{C}.\]
For $\varphi$ a normal positive functional on $\B(\H\otimes\H)$ (\textit{e.g.}\! for $\phi$ a normal state on $\B(\H)$ as above,
and $\varphi = \phi\comp\Ad_u$), the regularized functional $\varphi^E$ (here independent of $(u,\mathfrak{R}_u)$) is defined as \[ \varphi^E := x\in \B(\H\otimes\H) \mapsto \varphi\left((P_E\otimes \1)\,x\,(P_E\otimes \1) \right) \in \mathbb{C}.\]
\end{defi}

For a fixed normal state $\phi$ and a fixed pair $(u,\mathfrak{R}_u)$, the regularized functionals $\phi^{E,u}$ are normal positive contractions, and after normalization, $\phi^{E,u}/\|\phi^{E,u}\|$ are again normal states. As the state $\phi$ is normal, both $\phi^{E,u}$ and $\phi^{E,u}/\|\phi^{E,u}\|$ converge, as $E\to +\infty$, to the original state $\phi$ in the weak* topology. 
The same reasoning holds analogously for $\varphi$ a normal state on $\B(\H\otimes\H)$.
In the case of $\varphi = \phi\comp\Ad_u$, the restriction of $\phi^{E,u}$ to $\mathfrak{R}_u$ ``corresponds'' to the restriction of $\varphi^E$ to the first tensor component (denoted as $(\varphi^E)_1$), which in turn is equal to $\varphi_1(P_E\,\cdot\,P_E)$. The last converges to $\varphi_1$ in the weak* topology, as $E\to+\infty$.

\begin{defi}\label{def:cut2} Consider $\psi$ a state on $\B(\H)$. For $I\in\mathcal{J}$, $\delta>0$ (with $I_\delta\in\mathcal{J}$) and $E>0$, the regularized entropy $H_{I,\delta}^E$ of $\psi$ is defined by
  \[ H_{I,\delta}^E(\psi) := \inf_{(u,\mathfrak{R}_u)} \; \inf_{\phi} \; \frac{1}{\lambda_\phi}S_{\mathfrak{R}_u}\left(\phi^{E,u}/\|\phi^{E,u}\|\right). \]
  Here, the first infimum takes into account all pairs $(u,\mathfrak{R}_u)$ as in Definition \ref{def:ursplit}.
  The second infimum runs over all normal states $\phi$ over $\B(\H)$ to which there is a parameter $\lambda_\phi\in(0,1]$ such that $\phi^{E,u} \ge \lambda_\phi\,\psi^{E,u}$ holds when restricted to $\mathfrak{A}(I)\vee\mathfrak{A}(I_\delta)'$.
\end{defi}

\begin{prop}\label{pro:cut2} For a M\"obius covariant local net with the split property (Definition \ref{def:split}),
the quantity $H_{I,\delta}^E(\omega)$ with cutoff $E$ is finite, and independent of $\delta$. More precisely,
\[
 H_{I,\delta}^E(\omega) \le C_E\log C_E + S_E < +\infty,
\]
\[
 \text{ where } \left\{ 
  \begin{array}{l}{\displaystyle
    C_E = 2 \sup_{t \ge 0}\{  |f(t)| \} \,\sum_{N=0}^E \dim \ker (L_0 -N) }\\
    {\displaystyle S_E = 4 \sup_{t \ge 0}\{ |f(t)\log f(t)| \} \, \sum_{N=1}^E \dim\ker(L_0-N)}
  \end{array} \right.
\]
with $f$ an energy function as in Lemma \ref{lem:f}.
\end{prop}

\begin{proof}
Recall the functionals $\theta_{\delta,+}$ and $\tau_{\delta}$ defined by equations \eqref{eq:thetapm} and \eqref{eq:taud}, respectively.
The functional $\theta_{\delta,+}^E = \theta_{\delta,+} \left( (P_E\otimes \1)\,\cdot\,(P_E\otimes \1)\right) $, regularized as in Definition \ref{def:cutstate}, is a normal positive functional which converges to $\theta_{\delta,+}$ in the weak* topology, as $E\to\infty$.
Also, its restriction to the first tensor component is just $\tau_\delta\left(P_E\,\cdot\,P_E\right)$, which we denote by $\tau_{\delta,E}$.
Whereas $\theta_{\delta,+}$ and $\tau_{\delta}$ are only guaranteed to be well-defined if the net satisfies condition \ref{def:ea}(\ref{it:da}),
the regularized functionals $\theta_{\delta,+}^E$ and $\tau_{\delta,E}$ are well-defined even if the net only satisfied the split property.
And since $\theta_{\delta,+}^E\comp \Ad_u \ge \omega^E$, by arguments analog to Lemma \ref{lem:hh},
it follows that $\|\tau_{\delta,E}\|\Svn(\tau_{\delta,E}/\|\tau_{\delta,E}\|)$ is an upper bound for the regularized entropy given by Definition \ref{def:cut2}.
By these reasons, from here on we do not need the requirement of condition \ref{def:ea}(\ref{it:da}),
and only require the net to satisfy the split property.

We now focus on estimating $\Svn(\tau_{\delta,E}/\|\tau_{\delta,E}\|)$. From $\tau_\delta=\theta_{+,1}(\cdot \otimes \1)$ as expressed in equation \eqref{eq:taud} and considering that all $\Phi_n$ are eigenvectors of $L_0$ with eigenvalue $l_n$,
the only non-vanishing terms of $\tau_{\delta,E}$ are those corresponding to $\phi_{l,n}$ such that $l_n \le E$. One then has, after cutoff,
 \[
  \tau_{\delta,E}(x) = \omega+
  \sum_{n>0}^{l_n \le E} \;  \sum_{k=0}^3 \frac{|f_\delta(l_n)|}{2}\,\phi_{k,n}(x)
  \qquad (\, x\in \B(\H) \,).
 \]
The formula above allows us to use Corollary \ref{pro:Svnsum}, and therefore, the entropy $\Svn(\tau^E/\|\tau^E\|)$ of the normalized state can be estimated by the following:
 \[
  \Svn\left(\frac{\tau_{\delta,E}}{\|\tau_{\delta,E}\|}\right) \le
  \log( c_{\delta,E} ) + \frac{1}{c_{\delta,E}}{S_{\delta,E}},
 \]
where $c_{\delta,E}$ and $S_{\delta,E}$ are respectively the norm $\|\tau_{\delta,E}\|$ and the ``non normalized entropy'' defined by
\begin{align*}
  c_{\delta,E} &:= \|\tau_{\delta,E}\| = \sum_{N=0}^E 2\dim(\H_N) \, |f_\delta(N)|\\
  S_{\delta,E} &:= \sum_{N=1}^{E} 4\dim(\H_N) \; \left(-\frac{|f_\delta(N)|}2\log\frac{|f_\delta(N)|}2 \right) 
\end{align*}

As currently presented, the upper bound for $\Svn(\tau_{\delta,E}/\|\tau_{\delta,E}\|)$ still depends on $\delta$. However, $c_{\delta,E}$ and $S_{\delta,E}$ can be respectively bounded by constants $C_E$ and $S_E$ that are independent of $\delta$, given by the following:
\begin{align*}
  C_E &:= 2\,\|f\|_{\infty} \,\sum_{N=0}^E \dim(\H_N), \\
  S_E &:= 4\,\|f\log f\|_\infty \, \sum_{N=1}^E \dim(\H_N). 
\end{align*}

Therefore, joining the above to the bound of $\Svn(\tau_{\delta,E}/\|\tau_{\delta,E}\|)$, one finally has:
\[
  \Svn\left(\frac{\tau_{\delta,E}}{\|\tau_{\delta,E}\|}\right) \le
  \log C_E + \frac{S_E}{c_{\delta,E}} \le +\infty.
\]
Hence, by arguments analog to Lemma \ref{lem:hh}, and by noting that
$1\le c_{\delta,E} \le C_E$ and $S_{\delta,E} \le S_E$, we obtain
\[
 H_{I,\delta}^E(\omega) \le C_E\log C_E + S_E.
\]
The upper bound $C_E \log C_E + S_E$ above is finite and independent from $\delta$, thus proving the claim.
\end{proof}


Moreover, if condition \ref{def:ea}(\ref{it:da}) holds,
the bound can be more explicit.
The condition implies $d_N := \dim\ker(L_0-N) \le ce^{N^\kappa} \le ce^{N}$ for some $c>0$ and $\kappa\in(0,1)$.
Then $\sum_{N=0}^E e^{N} \le c' e^E$ for some $c'>0$. Thus, we have $H_{I,\delta}^E(\omega) < c'' E e^E$, for some $c''>0$.

\subsection{Lifting the regularization by distance}
Finally, we consider the quantity $H_I^E$ with cutoff $E$ a limit of the above when taking $\delta$ as a parameter approaching zero
(and hence $I_\delta$ approaching $I$).

\begin{defi}\label{def:cut3} Consider $\psi$ a normal state on $\B(\H)$.
For $I\in\mathcal{I}$ and a cutoff parameter $E>0$, we define:
 \[
  H_{I}^E(\psi) := \lim_{\delta\searrow 0}H_{I,\delta}^E(\psi)
 \]
where $H_{I,\delta}^E$ is as in Definition \ref{def:cut2}.
\end{defi}

\begin{theo}\label{thm:cut3} For a M\"obius covariant local net with the split property (Definition \ref{def:split}),
the quantity $H_I^E(\omega)$ can be bounded from above. More precisely,
\[ H_{I}^E(\omega) \le C_E\log C_E + S_E < +\infty, \]
\[ \text{ where } \left\{ 
  \begin{array}{l}{\displaystyle
    C_E = 2 \sup_{t \ge 0}\{ |f(t)| \} \,\sum_{N=0}^E \dim \ker (L_0 -N) }\\
    {\displaystyle S_E = 4 \sup_{t \ge 0}\{ |f(t)\log f(t)| \} \, \sum_{N=1}^E \dim\ker(L_0-N)}
  \end{array} \right. \]
with $f$ an energy function as in Lemma \ref{lem:f}.\end{theo}

The theorem is actually a mere corollary of Proposition \ref{pro:cut2}, but is indeed the main result of this work.
We have thus established the finiteness of the regularized entanglement entropy for a Möbius covariant local net satisfying the split property, that is, a relativistic chiral component of a quantum field in the algebraic setting.

\section{Conclusions and final remarks}\label{sec:conclusions}

In the present work, we focused our attention on chiral components of two-dimensional CFT, namely the M\"obius covariant local nets.
Provided the split property holds, we have given a sensible definition for regularized entropic quantities restricted
to an interval $I\in\mathcal{I}$. Considering the vacuum state, we also provided an upper bound
with a conformal energy cutoff $E$.

We recapitulate our definitions and comment on them a little further. Taking an interval $I$ and a small separation parameter $\delta$,
we consider all intermediate pairs $(u,\R_u)$ between $\A(I)$ and $\A(I_\delta)$, and all states $\phi$ that majorize $\omega$
when restricted to $\A(I)\vee\A(I_\delta')$, that is, everywhere besides a vicinity of the boundary of the intervals $I$ and $I_\delta'$.
The quantity $H_{I,\delta}$ regularized by $\delta$ (Definition \ref{def:cut1}) considers then the infimum of all entropies
of the states $\phi$ on $\R_u$.
With a cutoff $E$, our quantity $H_{I,\delta}^E$ (Definition \ref{def:cut2}) considers the infimum of all entropies of $\phi$
restricted to $\R_u$ but ``adjoined'' by the projection $u^*(P_E\otimes \1)u$.
Lastly, the quantity $H_I^E$ with cutoff $E$ is obtained by the limit $\delta \searrow 0$ (Definition \ref{def:cut3}).
As $\delta$ approaches zero, the local algebras $\A(I)$ and $\A(I_\delta)$ are very close to each other,
hence the quantity calculated here should reflect the property of entanglement between $\A(I)$ and its commutant $\A(I)' = \A(I')$
of the given state.
The infimum on the majorizing states, however, excludes those with too much aberrant behavior near the boundary,
in particular split states on $\A(I) \vee \A(I_\delta')$ are not counted.
Furthermore, we can avoid the trouble with type III local algebras by considering intermediate type I factors.
Our cutoff with respect to the conformal Hamiltonian $L_0$ is implemented naturally by applying
$P_E\otimes \1$ to the state, where $\H$ and $\H\otimes \H$ are identified by the intermediate type I factor
between $\A(I) \subset \A(I_\delta)$ and the spectral projection $P_E$ of $L_0$.
We obtained indeed that our $H_I^E$ is finite with some estimate as a function of $E$.

Let us compare our results with the physics literature.
Considering the $U(1)$-current model, our estimates are of the order of $E e^E$, which have a much worse divergence
than the estimates $\frac13\log(l/a)$ of Holzhey-Larsen-Wilczek \cite{hlw94,calcar04}, where $l$ is the length of the interval
(in the real line picture) and $a$ is the lattice spacing, hence $\frac1a$ should correspond to the energy cutoff $E$.
Not only that, our results do not even display a dependence on the interval length $l$.
The technical reason for such aspects of our result is that our estimates depend only on an orthonormal basis
of eigenvectors of the conformal Hamiltonian, a ``very global'' operator.
Here, sharper estimates ought to take into consideration the characteristics of each local algebra $\A(I)$, to bring a bound dependent in $l$.
Yet, also in another operator-algebraic work on entanglement entropy \cite{HS17}, the $\log(l/a)$-dependence could not be obtained.
Actually, in general, there are several possible definitions of entanglement entropy
which coincide with each other when the state is pure (see e.g.\! \cite[Theorem 3]{VP98}).
It is unclear to which definition the lattice approach corresponds.
This suggests that the expression of the entanglement entropy in the physics literature is specific to the lattice regularization,
and it is difficult to reproduce it directly in the continuum.

Another interesting question is how to adapt the methods provided here to theories in higher dimensions,
and implement cutoff with respect to the Hamiltonian $H$.
In contrast to the conformal Hamiltonian $L_0$, the usual Hamiltonian $H$ does not have a discrete spectrum,
yet energy nuclearity conditions and some ideas from the present paper might help defining an appropriate cutoff.
As the energy nuclearity index contains a natural dependence on the size of the region,
a successful approach should lead to an estimate of entropy depending also the region.
Furthermore, the behaviour of $\dim (L_0 - N)$ is related with the central charge of a conformal net
if one assumes modularity \cite{KL05}, hence the central charge $c$ appears in a natural way, in accordance with the physics literature.

A further different approach would be as in \cite{narnhofer94}, where a quantity that should correspond to our $H_{I,\delta}(\omega)$ is defined,
and expected to be finite (the proof required a $p$-nuclearity condition, and implicitly assumed the concavity of the one-subalgebra entropy
of Connes-Narnhofer-Thirring). In fact, at finite separation $\delta$, this might even be preferable to our Definition \ref{def:cut1},
since it does not have to consider all pairs $(u,\R_u)$ of the split property. The problem in this setting
is how to include a true ``energy cutoff'' that tames the divergence in $\delta$. 


\subsubsection*{Acknowledgments}
Y.O.\!  wishes to express his gratitude to Yasuyuki Kawahigashi for his continuous support and incentive, and to Narutaka Ozawa for his hospitality at RIMS.
Y.T.\! thanks Daniela Cadamuro, Roberto Longo and Ko Sanders for stimulating discussions.

Y.O.\! was supported by the Kokuhi-ryugaku scholarship of MEXT, Japan; Leading Graduate Course for Frontiers of Mathematical Sciences and Physics; and Hakushi Katei Kenkyu Suikou Kyouiku Seido of The University of Tokyo.
Y.T.\! was supported by the JSPS overseas research fellowship.


\end{document}